\documentclass[sigconf,nonacm]{acmart}

\usepackage{amsmath}
\usepackage{amsthm}
\usepackage{fullpage}
\usepackage{dcolumn}
\usepackage{booktabs}
\usepackage{tikz}
\usepackage{comment}
\usepackage{graphicx}
\usepackage[ruled, noend, linesnumbered]{algorithm2e}
\usepackage{mdframed}
\usepackage{enumitem}

\newtheorem{definition}{Definition}
\newtheorem{theorem}{Theorem}

\renewcommand{\paragraph}[1]{\vspace*{1mm}\noindent\textbf{#1}}

\newcommand{\ourmech}{Tabular-DDP Mechanism\xspace}

\usetikzlibrary{positioning,shapes,arrows}

\newcolumntype{M}[1]{D{.}{.}{1.#1}}

\title{Improving Utility for Privacy-Preserving Analysis of Correlated Columns using Pufferfish Privacy}

\begin{abstract}
Surveys are an important tool for many areas of social science research, but privacy concerns can complicate the collection and analysis of survey data. Differentially private analyses of survey data can address these concerns, but at the cost of accuracy---especially for high-dimensional statistics. We present a novel privacy mechanism, the \ourmech, designed for high-dimensional statistics with incomplete correlation. The \ourmech satisfies dependent differential privacy, a variant of Pufferfish privacy; it works by building a causal model of the sensitive data, then calibrating noise to the level of correlation between statistics. An empirical evaluation on survey data shows that the \ourmech can significantly improve accuracy over the Laplace mechanism.
\end{abstract}

\author{Krystal Maughan}
\email{Krystal.Maughan@uvm.edu}
\affiliation{\institution{University of Vermont}\country{USA}}

\author{Joseph P. Near}
\email{jnear@uvm.edu}
\affiliation{\institution{University of Vermont}\country{USA}}

\begin{document}
\settopmatter{printfolios=true}
\maketitle


\section{Introduction}

Survey data remains an important part of research in many different areas, including political science~\cite{sturgis2021demise}. Many survey questions are about political or otherwise personal beliefs or intentions, and individuals will rightfully be concerned if their responses may be made public. This concern even has the potential to reduce participation, which may bias the survey results. To address this problem, survey researchers typically keep their datasets secret in order to protect the privacy of respondents, and take additional steps to protect privacy when revealing aggregate results. These practices make it difficult to share survey data with other researchers, and in spite of the steps taken to protect privacy, respondents often remain concerned about the privacy of their responses.

Differential privacy~\cite{dwork2006calibrating, dwork2014algorithmic} is a strong formal definition of individual privacy, and it has been previously applied to survey data to protect the privacy of respondents~\cite{evans2022differentially}. Differential privacy works by adding noise to results destined for public release. Releasing more results requires adding more noise, because of the potential for correlation between results to reveal more information about a respondent than any single result does on its own. In differential privacy, this principle is called \emph{sequential composition}.

For survey researchers, sequential composition means that the error in the differentially private statistics they release increases with the number of statistics. For summary statistics about per-question responses, the error can grow large for long surveys with many questions.

We propose a novel mechanism for releasing differentially private statistics, the \ourmech, that can significantly improve error for releases of multiple statistics---including summary statistics about survey results. The key insight of the \ourmech is that a single respondent's answers to different survey questions are \emph{not} necessarily 100\% correlated, so the amount of noise required to use sequential composition is larger than necessary.

The \ourmech works by building an approximate causal model of the distribution underlying the collected survey data, then using the model to estimate correlations between statistics in the desired data release. The mechanism leverages incomplete correlations (and independence) to reduce the amount of noise required, based on a relaxed privacy definition called dependent differential privacy~\cite{liu2016dependence}.

In this paper, we formalize the \ourmech and prove that it satisfies dependent differential privacy. Then, we apply the \ourmech to real-world survey data from the American National Election Studies (ANES). We conduct an empirical evaluation of the accuracy of the \ourmech; the results suggest that the \ourmech can improve accuracy for summary statistics for this kind of survey data by several times in comparison to the standard Laplace mechanism (with sequential composition).

\paragraph{Contributions.} We make the following contributions:
\begin{itemize}[leftmargin=14pt]
\item We initiate the study of optimal mechanisms for differentially private summary statistics for survey results, based on the insight that responses are not completely correlated
\item We define the \ourmech, a novel dependent differential privacy mechanism designed for incompletely-correlated high-dimensional statistics
\item We evaluate the \ourmech experimentally using real survey data to demonstrate its accuracy benefit
\end{itemize}

\section{Background}

\subsection{Survey Data}

\begin{figure}
  \begin{mdframed}
    \begin{enumerate}[leftmargin=7pt]
    \item First, how much do you think people can change the kind of
      person they are? \\ \textit{$\;\;\;\bigcirc$ Completely $\;\;\;\bigcirc$ A
        lot $\;\;\;\bigcirc$ A moderate amount \\ $\;\;\;\bigcirc$ A little
        $\;\;\;\bigcirc$ Not at all}\\[1pt]
    \item If you wanted to defend an opinion of yours, how
      successfully do you think you could do that?\\
      \textit{$\;\;\;\bigcirc$ Extremely successfully $\;\;\;\bigcirc$ Very
        successfully \\$\;\;\;\bigcirc$ Moderately successfully $\;\;\;\bigcirc$
        Slightly successfully \\$\;\;\;\bigcirc$ Not successfully at all?}
    \end{enumerate}
  \end{mdframed}
  \caption{Example questions and responses from the ANES 2006 Survey.
    The survey has a total of 72 questions.}
  \label{fig:example_survey}
\end{figure}

The motivating use case for our work is privacy in survey data. Such data is collected by posing survey questions like the examples in Figure~\ref{fig:example_survey} to individuals, and aggregating and analyzing the responses. To protect privacy, the responses themselves are typically kept secret; even summary statistics about the responses are often not released publicly, because they could potentially reveal information about individual respondents.

The standard approach for protecting privacy in survey data is \emph{de-identification}: the removal of \emph{personally identifiable information} (PII) like names and phone numbers~\cite{connors2019transparency, plutzer2019privacy} before sharing the data. However, de-identification approaches do not always fully protect privacy: they are frequently subject to re-identification attacks~\cite{henriksen2016re}, which recover the removed PII. In addition, aggressive de-identification can remove useful information from the data, reducing utility.

We focus on producing privacy-preserving histograms of response counts for each question (i.e. for each question, how many respondents chose each possible response for that question), with a formal privacy guarantee. Based on this goal, the number of statistics we want to release grows linearly with the number of questions in the survey.

\subsection{Differential Privacy}

Differential privacy~\cite{dwork2006calibrating, dwork2014algorithmic}
is a formal privacy definition based on the notion of
indistinguishability. Informally, for every hypothetical individual
who \emph{could} contribute data for an analysis, differential privacy
ensures that the analysis results will not reveal whether or not the
individual \emph{did} contribute data.
\begin{definition} [Differential Privacy]
  A randomized \emph{mechanism} $\mathcal{M}$ satisfies
  $(\epsilon,\delta)$-differential privacy if, for all datasets $D$
  and $D'$ that differ in the data of one individual, and all possible
  sets of outcomes $S$:
  \[
    \Pr[\mathcal{M}(D)\in \mathcal{S}] \leq e^{\epsilon}
    \Pr[\mathcal{M}(D')\in \mathcal{S}] + \delta
  \]
\end{definition}
Differential privacy is \emph{compositional}: if $\mathcal{M}_1$
satisfies $(\epsilon_1, \delta_1)$-differential privacy, and
$\mathcal{M}_2$ satisfies $(\epsilon_2, \delta_2)$-differential
privacy, then releasing the results of both mechanisms satisfies
$(\epsilon_1 + \epsilon_2, \delta_1 + \delta_2)$-differential privacy.
Differential privacy is closed under \emph{post-processing}: if
$\mathcal{M}$ satisfies $(\epsilon, \delta)$-differential privacy,
then $f \circ \mathcal{M}$ satisfies $(\epsilon, \delta)$-differential
privacy for any function $f$.

Differential privacy is defined in terms of \emph{neighboring
  databases} that differ in the data of one individual. The formal
specification of this idea makes a big difference to the privacy
guarantee obtained in practice. The standard
approach~\cite{dwork2014algorithmic} is to assume that each individual
contributes exactly one row to the database, so the \emph{distance}
between two databases is equal to the number of rows on which they
differ.
When one individual may contribute multiple rows, a different distance
metric must be used to ensure privacy.

To achieve differential privacy, we can add noise as prescribed by one
of several basic mechanisms. The two most commonly-used mechanisms are
the \emph{Laplace mechanism}, which ensures pure
$\epsilon$-differential privacy, and the \emph{Gaussian mechanism},
which ensures $(\epsilon, \delta)$-differential privacy. In both
cases, the scale of the noise is determined by the query's
\emph{sensitivity}, which measures the influence of a single
individual's data on the query's output. The \emph{$L1$ sensitivity}
of a function $f:\mathcal{D} \rightarrow \mathbb{R}^k$ is defined as
follows, where $d$ is a distance metric on databases:
\[ \Delta_1 f = \max_{D, D' . d(D, D') \leq 1} \lVert f(D) - f(D') \rVert_1 \]
The \emph{$L2$ sensitivity} $\Delta_2 f$ is defined the same way, but
with the $L2$ norm instead of the $L1$ norm.

\begin{theorem}[The Laplace Mechanism] Given a numeric query
  $f:\mathcal{D} \rightarrow \mathbb{R}^k$, the Laplace mechanism adds
  to the query answer $f(D)$ with a vector $(\eta_1,\cdots,\eta_k)$,
  where $\eta_i$ are i.i.d. random variables drawn from the Laplace
  distribution centred at 0 with scale $b=\Delta_1 f/\epsilon$,
  denoted by $Lap(b)$. The Laplace mechanism preserves
  $(\epsilon,0)$-differential privacy.
\end{theorem}


\subsection{Dependent Differential Privacy}

Sometimes, \emph{correlations may exist between individuals} that
allow an adversary to make inferences about one individual based on
the data of another. Consider, for example, a dataset of GPS locations
that includes members of a chess club. If the chess club meets at 3pm
on Thursdays, then the locations of the club's members at that time
will be highly correlated with one another! The adversary may be able
to learn the \emph{most popular} location of chess club members during
the meeting time, and then \emph{infer}, based on their belief about
correlations in the data, that an \emph{individual} chess club member
is highly likely to have been at the popular location. In this case,
the correlation in the data enabled the inference: absent the
knowledge that chess club members are likely to be in the same
location during the meeting time, the adversary would not be able to
make the inference.

Importantly, \emph{differential privacy does not promise to prevent
  this inference}. Arguably, it is not a privacy violation at all.
However, in some cases such inferences are highly likely to reveal
information that may prove harmful, so a significant body of work has
investigated ways of refining the definition of differential privacy
to account for this risk~\cite{song2017pufferfish, liu2016dependence,
  niu2019making, kessler2015deploying, liang2020pufferfish,
  zhang2022attribute}.

The most important for our setting is \emph{dependent differential
  privacy}, due to Liu et al.~\cite{liu2016dependence}. Dependent
differential privacy can be seen as a strengthening of differential
privacy, which reduces to differential privacy when no correlations
are present in the data. Dependent differential privacy is defined as
follows:
\begin{definition}[Dependent Neighboring Databases]
  Two databases $D(L, \mathcal{R})$ and $D'(L, \mathcal{R})$ are
  dependent neighboring databases if the modification of a tuple value
  in database $D(L, \mathcal{R})$ causes a change in at most $L-1$
  other tuple values in $D'(L, \mathcal{R})$ due to the probabilistic
  dependence relationship $\mathcal{R}$ between the data tuples.
\end{definition}
\begin{definition}[Dependent Differential Privacy]
  A randomized mechanism $\mathcal{M}$ satisfies $(\epsilon,
  \delta)$-dependent differential privacy if for all pairs of
  dependent neighboring databases $D(L, \mathcal{R})$ and $D'(L,
  \mathcal{R})$ and all possible sets of outcomes $S$:
  \[
    \Pr[\mathcal{M}(D(L, \mathcal{R}))\in \mathcal{S}] \leq e^{\epsilon}
    \Pr[\mathcal{M}(D'(L, \mathcal{R}))\in \mathcal{S}] + \delta
  \]
\end{definition}
This definition is designed to capture inferences made on the
dependence relationship $\mathcal{R}$ while preserving important
properties of differential privacy. Like differential privacy,
dependent differential privacy is compositional and closed under
post-processing.

Liu et al.~\cite{liu2016dependence} propose a definition of
\emph{dependent sensitivity} that allows the use of the Laplace
mechanism to satisfy dependent differential privacy. Dependent
sensitivity is large when significant correlations in the data could
enable inferences like our earlier example, and is equal to $L1$
sensitivity when no correlations exist.
\begin{definition}[Dependent sensitivity~\cite{liu2016dependence}]
  The dependent sensitivity of a query $Q$ with $L1$ sensitivity
  $\Delta Q$ is:
  \[ DS^Q = \sum_{j = C_{i1}}^{C_{iL}} \rho_{i j} \Delta Q \]
  Where $\rho_{i j}$ represents the \emph{dependence coefficient}
  between records $i$ and $j$.
\end{definition}

\section{Privacy for Survey Data}

Differential privacy assumes complete correlation between the
attributes of a single individual, and so releasing statistics about
multiple columns of a tabular dataset requires the use of sequential
composition. The key insight of our approach is the observation that
complete correlation often \emph{does not} exist between attributes,
so the use of sequential composition provides very loose upper bounds
on the actual privacy loss for these statistics.

We propose the use of a dependent differential privacy mechanism for
releasing statistics about multiple attributes in tabular data,
including survey data. Under valid assumptions about the distribution
of the underlying data, dependent differential privacy provides strong
privacy protection for participants in the dataset---but with less
noise required.

The primary challenges lie in modeling correlations between columns
and in efficiently calculating the dependent sensitivity of queries
over the data based on these models.


\begin{figure}
  \centering
  \begin{tabular}{c c c}
    \begin{minipage}{90pt}
      \centering
  {\footnotesize
    \begin{tabular}{|c |c| c|}
      \hline
      Prize & First & Monty \\
      Door & Selection & Opens \\
      \hline
      1 & 1 & 2 \\
      3 & 2 & 1 \\
      3 & 3 & 2 \\
      \hline
    \end{tabular}
  }
  \textbf{Original Table}
  \end{minipage}
  & $\Rightarrow$ &
  \begin{minipage}{60pt}
    \centering
    {\footnotesize
    \begin{tabular}{|c|}
      \hline
      Prize Door\\
      \hline
      1\\
      3\\
      3\\
      \hline
      First Selection \\
      \hline
      1\\
      2\\
      3\\
      \hline
      Monty Opens \\
      \hline
      2\\
      1\\
      2\\
      \hline
    \end{tabular}
  }
  \textbf{Transformed Table}
  \end{minipage}\\
  \end{tabular}

  \caption{Example Tabular Data: Records of Monty Hall Games.}
  \label{fig:monty_data}
\end{figure}

\subsection{Example: Monty Hall}

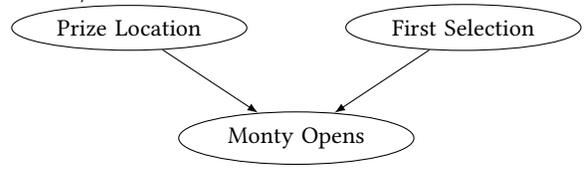
\begin{figure}
  \centering
\begin{tikzpicture}[
  node distance=1cm and 0cm,
  mynode/.style={draw,ellipse,text width=2cm,align=center}
]
\node[mynode] (sp) {Prize Location};
\node[mynode,below right=of sp] (gw) {Monty Opens};
\node[mynode,above right=of gw] (ra) {First Selection};
\path (sp) edge[-latex] (gw) 
(gw) edge[latex-] (ra);
\end{tikzpicture}

\caption{Bayesian Network for the Monty Hall Problem}
\label{fig:bayesian_network}
\end{figure}

As a simple example of our setting, consider the Monty Hall
problem. The problem describes a game involving a
contestant, a host (Monty Hall), and three doors. One door contains a
goat, one contains a prize, and one is empty; the contestant's goal is
to choose the door with the prize. The game proceeds in three steps:

\begin{enumerate}
\item The contestant chooses a door (the ``First Selection'').
\item Monty opens a door that is \emph{neither} the ``First
  Selection'' \emph{nor} the door with the prize (revealing either the
  goat or nothing at all).
\item The contestant is given the opportunity to change their
  selection to the other non-open door, or keep their first selection.
\item The contestant's final selection is opened. If the door contains
  the prize, the contestant wins.
\end{enumerate}

The Bayesian network corresponding to the Monty Hall problem appears
in Figure~\ref{fig:bayesian_network}. This problem is famous for being
counterintuitive---we assume that the event of Monty opening one of
the doors does not affect the probability that the contestant has made
the right choice, but in fact it does! This effect is encoded in the
Bayesian network: which door Monty opens depends on both the location
of the prize \emph{and} the contestant's first selection.

Imagine we have collected observations of Monty Hall games, as in
Figure~\ref{fig:monty_data}, and we would like to release statistics
about these games under differential privacy. We can release
histograms for all three attributes summarizing the game outcomes, and
add Laplace noise with scale $\frac{1}{\epsilon}$ to each one.
By the sequential proposition property of differential privacy, the
total privacy cost is $3\epsilon$. Note that it is not possible to use
parallel composition in this case, because adding or removing a whole
row of data changes the results of all three histograms.

\subsection{Modeling Correlations}
\label{sec:model-corr}

Calculating dependent sensitivity requires the ability to evaluate the
probability that an attribute takes a particular value given the
values of the other attributes in the same row. We model these
correlations using a Bayesian network, in a similar way to
previous work~\cite{liu2016dependence, song2017pufferfish}.

\begin{figure}
  \centering

\begin{tikzpicture}[
  node distance=.6cm and .5cm,
  mynode/.style={draw,ellipse,text width=1cm,align=center}
]
\node[mynode] (3) {V06P431};
\node[mynode, below=of 3] (4) {V06P432};
\node[mynode, below left=of 4] (5) {V06P433};
\node[mynode, below right=of 4] (6) {V06P434};
\node[mynode, below=of 5] (8) {V06P510};
\node[mynode, below=of 6] (7) {V06P505};
\path (3) edge[-latex] (4) 
(4) edge[-latex] (5)
(4) edge[-latex] (6)
(5) edge[-latex] (8)
(6) edge[-latex] (7);
\end{tikzpicture}

\caption{Example Bayesian network for a subset of ANES 2006 Survey
  Data. Each node represents one column in the original dataset; each
  edge is associated with a conditional probability table encoding the
  conditional dependencies between column values.}
  \label{fig:ex_network}
\end{figure}
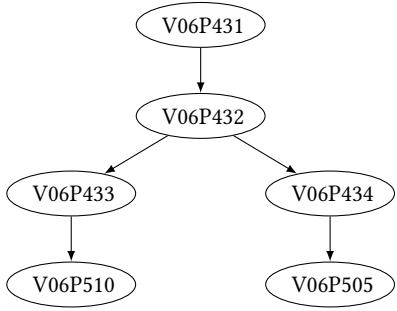

A Bayesian network is a graphical model (directed acyclic graph) that
represents conditional dependencies between variables. In our setting,
each column of the dataset is represented by a variable in the
Bayesian network (i.e. a node in the graph) and the conditional
probability table associated with each edge in the graph encodes the
conditional dependencies between column values.

We represent a Bayesian network learned from the dataset using a
triple $(V, E, P)$, where $V$ and $E$ are the vertices and edges of
the graph, respectively, and $P$ is the conditional probability table.
For every pair of attributes $X_1, X_2 \in V$, if $X_1$ is
conditionally dependent on $X_2$, then an edge $(X_1, X_2) \in E$ will
connect them, and the conditional probability table will record the
corresponding conditional probability distribution: for every possible
value $v_1$ and $v_2$ that attributes $X_1$ and $X_2$ could take,
$P(X_1=v_1, X_2=v_2) = \Pr[X_2 = v_2 \mid X_1 = v_1]$.

In a survey, we expect that the attributes of a single individual's
results will be correlated with each other. We model the extent of
this correlation using a Bayesian network, so that we can apply
mechanisms for dependent differential privacy (described in
Section~\ref{sec:depend-sens-tabul}).

\subsection{Learning the Model}

The major challenge of this approach is defining the Bayesian network
itself. Previous work has assumed that the network is already known,
and is public information~\cite{liu2016dependence, song2017pufferfish}
(and often, that it has a specific form---e.g. a Markov chain).

Our approach is to learn the Bayesian network from the data itself.
Learning the structure of Bayesian networks from data is a challenging
but well-studied problem~\cite{scanagatta2019survey,
  tsamardinos2006max}; our implementation uses the Pomegranate library
for Python.

An example Bayesian network learned from a subset of the columns of
the ANES 2006 Survey dataset appears in Figure~\ref{fig:ex_network}.
Approaches for learning structures like these do not scale well to
large networks (e.g. hundreds of attributes---as is common in
surveys). In order to make the model-learning component of our
approach tractable, we split the attributes into smaller chunks (in
our evaluation, we include 10 attributes per chunk), and learn a model
for just the attributes in each chunk. Then, to provide privacy for
the whole response, we add noise to each chunk separately and use the
sequential composition property to determine the total privacy loss.

\subsection{Privacy Considerations}
\label{sec:priv-cons}

The approach we have outlined raises several important concerns about
the real-world privacy we can expect from the guarantee. First,
Pufferfish privacy and its variants (including dependent differential
privacy) represent weaker guarantees than $\epsilon$-differential
privacy; in the context of survey data, the weakening of the guarantee
is similar to the difference between node- and edge-level privacy in
graphs~\cite{kasiviswanathan2013analyzing}. In our setting:
\begin{itemize}
\item \textbf{$\epsilon$-differential privacy} protects the presence
  or absence of \emph{one individual} in the survey results
\item \textbf{$\epsilon$-dependent differential privacy} protects the
  presence or absence of \emph{one answer to a survey question} in the
  survey results
\end{itemize}
The difference between these guarantees is significant, and our weaker
guarantee may not be applicable in some cases. In cases where survey
answers may be sensitive, but participation in the survey is not, the
dependent differential privacy guarantee may be appropriate, and
enable better utility in the results.

Second, learning a Bayesian network from the sensitive data presents
two additional concerns: (1) the model's structure may reveal
properties of the underlying distribution (e.g. enabling attribute
inference), and (2) the model's structure may reveal properties of
individual records in the data (enabling inferences about
individuals). In our setting, (1) is not a major concern, since the
underlying distribution of responses is what we would like to learn.

However, concern (2) is an issue in our setting. It is possible that
learning the Bayesian network from the data could reveal information
specific to individuals---though in large datasets, this information
is likely to be minimal. To alleviate this issue, a differentially
private learning algorithm could be used~\cite{zhang2017privbayes}.

An additional concern is that the learning process could produce a
model that does not actually match the underlying
distribution---either because the learning process fails to learn the
correct model, or because the data does not represent the underlying
distribution very well. In this case---as in other applications of
Pufferfish privacy---unexpected privacy failures could occur due to
the mismatch between \emph{expected} and \emph{actual} correlations in
the data.

All of these concerns represent limitations of our approach, and are
important areas for future improvement.



\section{Dependent Sensitivity for Tabular Data}
\label{sec:depend-sens-tabul}

  

  


\begin{algorithm}[t]
  \SetKwData{count}{count}
  \SetKwData{Lap}{Lap}
  \SetKwData{noisyCount}{noisyCount}
  \SetKwData{total}{total}
  \SetKwData{pr}{Prob}
  \SetKwInOut{Input}{Input}
  \SetKwInOut{Output}{Output}
  
  \let\oldnl\nl
  \newcommand{\nonl}{\renewcommand{\nl}{\let\nl\oldnl}}

  \Input{\hspace{1pt}Database $D$ with $n$ columns, query $Q$ to be
    run on each column, chunk size $k$, privacy parameter $\epsilon$}
  
  \Output{\hspace{1pt}Privacy-preserving statistics for each column}

  $\{D_1, \dots, D_{\lfloor {n}/{k} \rfloor} \} \leftarrow \textsc{SplitColumns}(D, k)$\\
  \For{$D_i \in \{D_1, \dots, D_{\lfloor {n}/{k} \rfloor} \}$}{
  $(V, E, P) \leftarrow \textsc{LearnNetwork}(D_i)$\\
  \For{$(X_i, X_j) \in E$}{
    $\rho_{i,j} \leftarrow \max_{d_j, d_{i_1}, d_{i_2}} 
    \log \Big( \frac{\Pr[X_i = d_{i_1}, X_j = d_j]} {\Pr[X_i = d_{i_2}, X_j = d_j]} \Big)$
  }
  $DS \leftarrow {\sum_{i,j} \rho_{i,j}}$
  \hfill \textit{calculate dependent sens.}\\
  \For{$X_i \in \textit{columns}(D_i)$}{
    $R_i \leftarrow Q(X_i) + \textsf{Lap}\Big( \frac{n DS}{k \epsilon} \Big)$
    \hfill \textit{calculate noisy result}\\
  }
  }
  \Return{$\textbf{R}$}\\[10pt]
  \nonl $\textsc{SplitColumns}(D, k)$ splits dataset $D$ column-wise into chunks, so that each chunk has at most $k$ columns.\\
  \nonl $\textsc{LearnNetwork}(D_i)$ learns a causal model for dataset $D_i$, expressed as a Bayesian network.\\
  \caption{The Tabular-DDP Mechanism.}
  \label{alg:mechanism}
\end{algorithm}

This section describes the \emph{\ourmech}, formalized in Algorithm~\ref{alg:mechanism}, which adapts the dependent sensitivity approach of Liu et al.~\cite{liu2016dependence} to the setting of multi-attribute tabular data.

\paragraph{Transforming the data.}
We adopt the definition of dependent sensitivity from Liu et al., as
defined earlier. To scale noise to dependent sensitivity, we need the
data to be represented in the form $X = \{X_1, \dots, X_n\}$, where we
assume that the attributes of each $X_i$ may be \emph{completely}
dependent on one another, and there may \emph{additionally} be
correlations between two tuples $X_i$ and $X_j$.

To fit these assumptions, we transform the tabular representation of
our data table into a single-column table, as shown in
Figure~\ref{fig:monty_data}, by concatenating the columns. After this
transformation, each tuple has only a single attribute, and the domain
of that attribute is the product of the table's original attributes.

The transformed data fits the assumptions of dependent sensitivity. In
the new representation, which has only a single column, the Bayesian
network in Figure~\ref{fig:bayesian_network} encodes correlations
between \emph{rows} rather than columns, as expected for dependent
sensitivity.

\paragraph{Calibrating noise to dependent sensitivity.}
With the transformed data, it is possible to apply the mechanisms of
Liu et al. directly:
\begin{enumerate}
\item Transform the tabular data to a single-column representation
\item Add Laplace noise to the results of querying the transformed
  data, scaled to the dependent sensitivity of the query
\end{enumerate}
Next, we introduce a slight modification to the mechanism that avoids
the need for explicit transformation of the data.

\paragraph{The \ourmech.}
The \ourmech, defined in Algorithm~\ref{alg:mechanism}, simulates the
process described above, and scales the additive Laplace noise to the
\emph{effective} dependent sensitivity of applying a query to multiple
attributes of a tabular dataset in parallel. 

First, the mechanism splits the dataset into chunks column-wise (line
1), to make the modeling task computationally tractable. Next, for
each chunk, the mechanism learns a Bayesian network encoding the
causal relationships in the data (line 3). The \textsc{LearnNetwork}
function refers to an off-the-shelf tool for learning the network and
returning a representation containing the conditional probability
table, as described earlier (Section~\ref{sec:model-corr}). The larger
the number of columns $k$ in each chunk, the more computationally
challenging this task is. Then, the mechanism computes the effective
dependent sensitivity by summing the dependence coefficients for all
attributes in the table (line 5). Here, the mechanism uses the
conditional probability table in the learned Bayesian network to
calculate the probability ratio:
\[\frac{\Pr[X_j = d_j | X_i = d_{i_1}]} {\Pr[X_j = d_j | X_i = d_{i_2}]} \]
Finally, the mechanism releases the
result of running the query $Q$ and adding Laplace noise scaled to the
dependent sensitivity (line 8).
Like the process defined above, the \ourmech satisfies
$\epsilon$-dependent differential privacy, as long as the learned
Bayesian network accurately represents the underlying data
distribution. For each column-wise chunk of the dataset, the mechanism
satisfies $\lfloor \frac{n}{k} \rfloor \epsilon$-dependent
differential privacy, for a total privacy cost bounded by
$\epsilon$-dependent differential privacy by sequential composition.

\paragraph{Privacy.}
To prove privacy for the \ourmech, we will view the dataset implicitly
in the single-column representation described above (with correlations
between tuples, rather than columns) and leverage the privacy result of Liu et al.~\cite{liu2016dependence}:
\begin{lemma}[Liu et al.~{\cite[Theorem 8]{liu2016dependence}}]
  The dependent sensitivity for publishing any query $Q$ over a
  dependent (correlated) dataset is
  \[DS^Q = \max_i DS_i^Q\]
  \label{lem:dep_sens}
\end{lemma}
Here, $i$ refers to a tuple index, and $DS_i^Q = \sum_i \rho_{i,j}
\Delta Q_j$ is the dependent sensitivity for the $i$th tuple. If we
can show that the \ourmech correctly calculates $DS^Q$ and adds
Laplace noise scaled to that sensitivity, then it follows that the
\ourmech satisfies dependent differential privacy.
\begin{theorem}
  If the learned Bayesian network $(V, E, P)$ accurately represents the underlying distribution of the dataset $D$, then the Tabular-DDP mechanism (Algorithm~\ref{alg:mechanism}) satisfies
  $\epsilon$-dependent differential privacy.
\end{theorem}

\begin{proof}
  We show that Algorithm~\ref{alg:mechanism} satisfies $\frac{k\epsilon}{n}$-dependent differential privacy for each chunk of columns. By sequential composition, if there are at most $\frac{n}{k}$ chunks, then the mechanism has a total privacy cost of $\epsilon$-dependent differential privacy.
  For each chunk of columns, we have the following for the sensitivity
  calculated by Algorithm~\ref{alg:mechanism}, leveraging the fact
  that our counting queries have sensitivity $\Delta Q_j = 1$:
  \begin{align*}
    DS = & \sum_{i, j} \rho_{i,j}\\
    \geq & \max_i \sum_j \rho_{i,j} \Delta Q_j\\
    = & DS^Q
  \end{align*}
  By Lemma~\ref{lem:dep_sens}, noise scaled to $\frac{DS}{\epsilon}$
  will satisfy $\epsilon$-dependent differential privacy.
  Algorithm~\ref{alg:mechanism} adds Laplace noise scaled to:
  \[\frac{n DS}{k \epsilon}\]
  which satisfies $\frac{k\epsilon}{n}$-dependent differential
  privacy, as required.

\end{proof}

\paragraph{Utility.}
The same accuracy bounds proven by Liu et al.~\cite{liu2016dependence}
also apply to the \ourmech. These results bound the error for any individual column of the statistics returned by the mechanism.
\begin{definition}[$(\alpha, \beta)$-accuracy~\cite{dwork2014algorithmic, liu2016dependence}]
  A randomization algorithm $\mathcal{A}$ satisfies $(\alpha,
  \beta)$-accuracy for a query function $Q$ if:
  \[\Pr[\max_D \lvert \mathcal{A}(D) - Q(D) \rvert > \alpha] \leq  \beta\]
\end{definition}
\begin{lemma}
  The \ourmech provides $(\alpha, \beta)$-accuracy for each column of
  the dataset $D$, for $\beta = \exp \Big(\frac{-\epsilon \alpha}{DS^Q} \Big)$.
  \label{lem:accuracy_column}
\end{lemma}
\begin{proof}
  Follows directly from Liu et al.~\cite{liu2016dependence}, Theorem 10.
\end{proof}

In addition, we can extend the utility bounds from Liu et
al.~\cite{liu2016dependence} to bound $L1$ error for \ourmech. We leverage a result on the sum of Laplace samples from Chan et al.~\cite{chan2011private} (which adapts the Chernoff bound). We define $L_1$ accuracy in the same way as $(\alpha, \beta)$ accuracy, but using the $L_1$ error.
\begin{lemma}[Sum of independent Laplace samples ({\cite[Lemma 2.8]{chan2011private}})]
Suppose $\gamma_i$'s are independent random variables, where each $\gamma_i$ has Laplace distribution $\mathsf{Lap}(b_i)$. Suppose $Y := \sum_i \gamma_i$, and $b_M := \max_i b_i$. Let $\nu \geq \sqrt{\sum_i b_i^2}$ and $0 < \lambda < \frac{2\sqrt{2}\nu^2}{b_M}$. Then:
\[\Pr[Y > \lambda] \leq \exp\Big(\frac{-\lambda^2}{8\nu^2}\Big)\]
\label{lem:chernoff}
\end{lemma}
\begin{definition}[$L_1$-accuracy]
  A randomization algorithm $\mathcal{A}$ satisfies $L_1 (\alpha,
  \beta)$-accuracy for a query function $Q$ if:
  \[\Pr[\max_D \lVert \mathcal{A}(D) - Q(D) \rVert_1 > \alpha] \leq \beta\]
\end{definition}
%
\begin{theorem}
  The \ourmech provides $L_1 (\alpha, \beta)$-accuracy for $\beta = \exp\Big(\frac{- \sqrt{2}\epsilon\alpha}{4 DS}\Big)$.
\end{theorem}
To prove the accuracy bound, we consider that the $L_1$ error introduced by the mechanism is a result \emph{only} of the noise samples added to each result in line 8 of Algorithm~\ref{alg:mechanism} (i.e. $\lVert \mathcal{A}(D) - Q(D) \rVert_1$ is exactly equal to the sum of the noise samples added by the mechanism). Each of these noise samples is conditionally independent from the others, so Lemma~\ref{lem:chernoff} applies, and gives an upper bound on the $L_1$ error resulting from the noise.
\begin{proof}
Set $\lambda = \frac{2\epsilon\sqrt{2}\nu^2}{DS}$ and $\nu = \sqrt{n}\frac{DS}{\epsilon}$. By Lemma~\ref{lem:chernoff}, we have:
\begin{align*}
\Pr[\max_D \lVert \mathcal{A}(D) - Q(D) \rVert_1 > \alpha] 
  \leq & \exp\Big(\frac{-\alpha^2}{8\nu^2}\Big) \\
  = & \exp\Big(\frac{-\alpha\alpha}{8\nu^2}\Big) \\
  = & \exp\Big(\frac{-\alpha\frac{2\epsilon\sqrt{2}\nu^2}{DS}}{8\nu^2}\Big) \\
  = & \exp\Big(\frac{- \sqrt{2}\epsilon\alpha}{4 DS}\Big) \\
\end{align*} 
\end{proof}
Thus the accuracy of \ourmech is independent of the dimensionality of the statistic being released, except as encoded in the dependent sensitivity.

\paragraph{Limitations.}
Our approach has several important limitations. First, as discussed in Section~\ref{sec:priv-cons}, the privacy guarantee is strictly weaker than standard $\epsilon$-differential privacy, and additional unexpected privacy failures could occur if the learned Bayesian networks do not actually correspond to the underlying population distribution. Second, the \ourmech is based on Laplace noise, and uses $L_1$ sensitivity; for high-dimensional data, if $(\epsilon, \delta)$-differential privacy is sufficient, the Gaussian mechanism with $L_2$ sensitivity may produce better accuracy. We hope to extend the \ourmech to Gaussian noise with $L_2$ sensitivity in future work.

\section{Evaluation}

\begin{figure*}
  \centering
  \hspace*{-18mm}\includegraphics[width=1.2\textwidth]{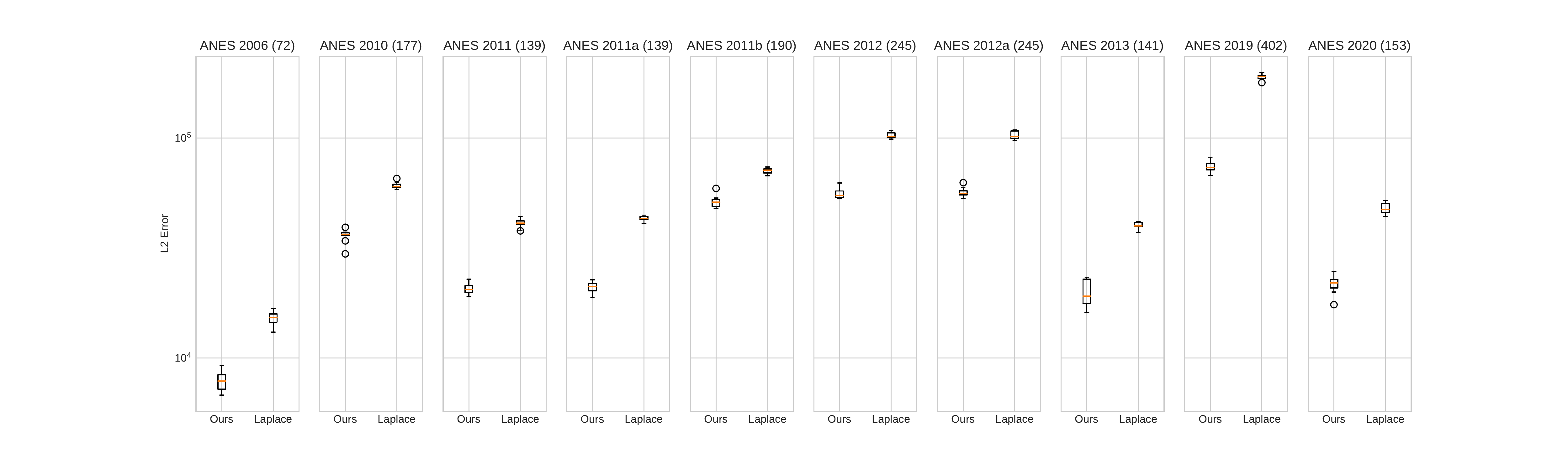}

  \caption{Accuracy results: comparison between the Tabular-DDP
    mechanism and the Laplace mechanism for $\epsilon=0.1$.}
  \label{fig:all_results}
\end{figure*}



\begin{figure}
  \centering
  \begin{tabular}{l cc}
    \hline
    \textbf{ANES Survey} & \textbf{Questions} & \textbf{Responses}\\
    \hline
    Pilot Study 2006 & 72 & 675 \\
    Eval. of Govt. \& Society 2010 & 117 & 1275 \\
    Eval. of Govt. \& Society 2011(a) & 139 & 1315 \\
    Eval. of Govt. \& Society 2011(b) & 139 & 1240 \\
    Eval. of Govt. \& Society 2012 & 190 & 1314 \\
    Pilot Study 2013 & 141 & 1635 \\
    Pilot Study 2019 & 402 & 3165 \\
    Pilot Study 2020 & 153 & 3080\\
    \hline
  \end{tabular}
  \caption{Evaluation Datasets.}
\end{figure}

\begin{figure}
  \centering
  \includegraphics[width=.45\textwidth]{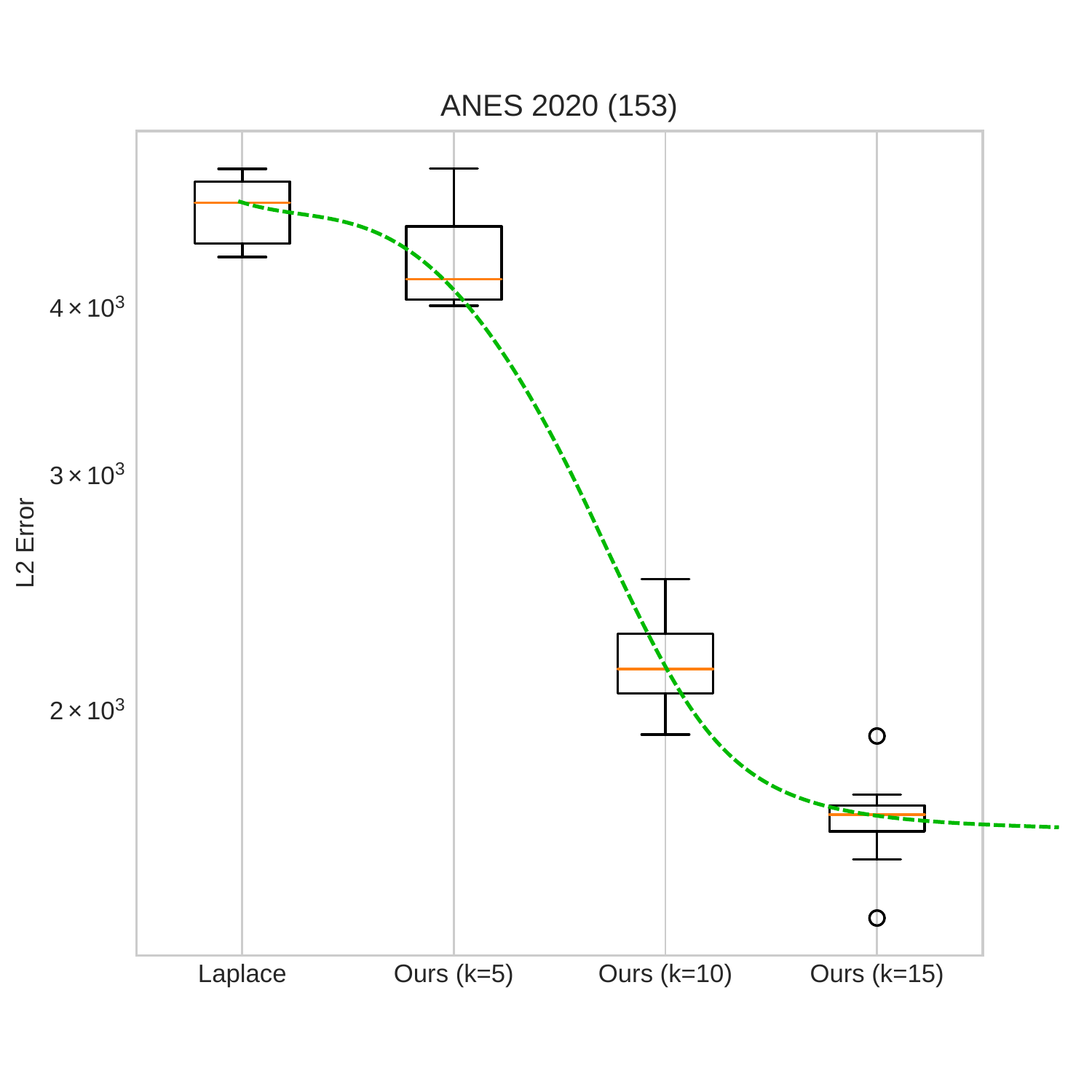}
  \caption{Accuracy results: effect of the chunk size $k$ for $\epsilon=1.0$. The green dashed line is a rough representation of the trend in $k$'s effect on accuracy.}
  \label{fig:k_results_results}
\end{figure}

\begin{figure*}
  \centering
  \begin{tabular}{c c}
    \includegraphics[width=.45\textwidth]{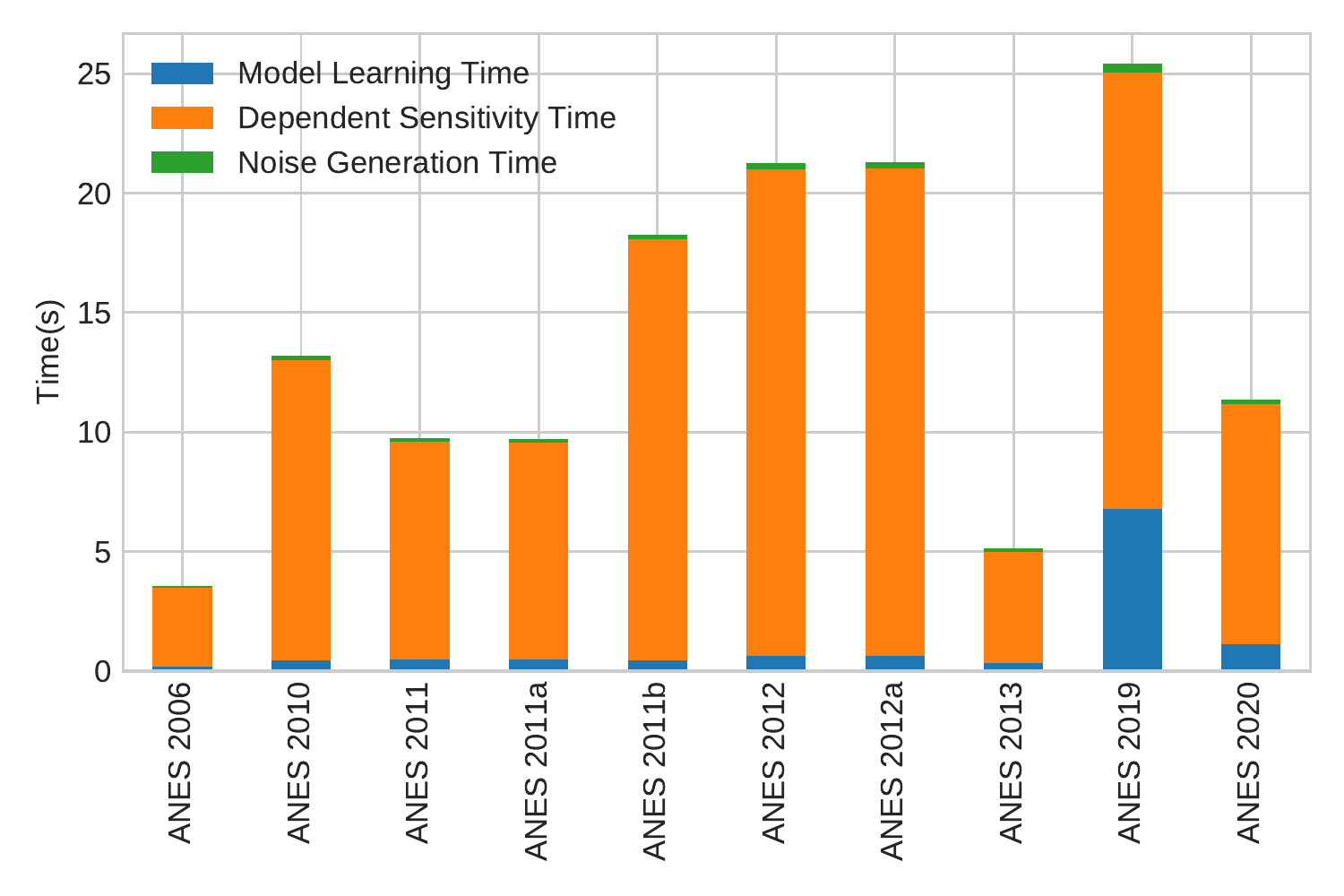}
    &
    \includegraphics[width=.45\textwidth]{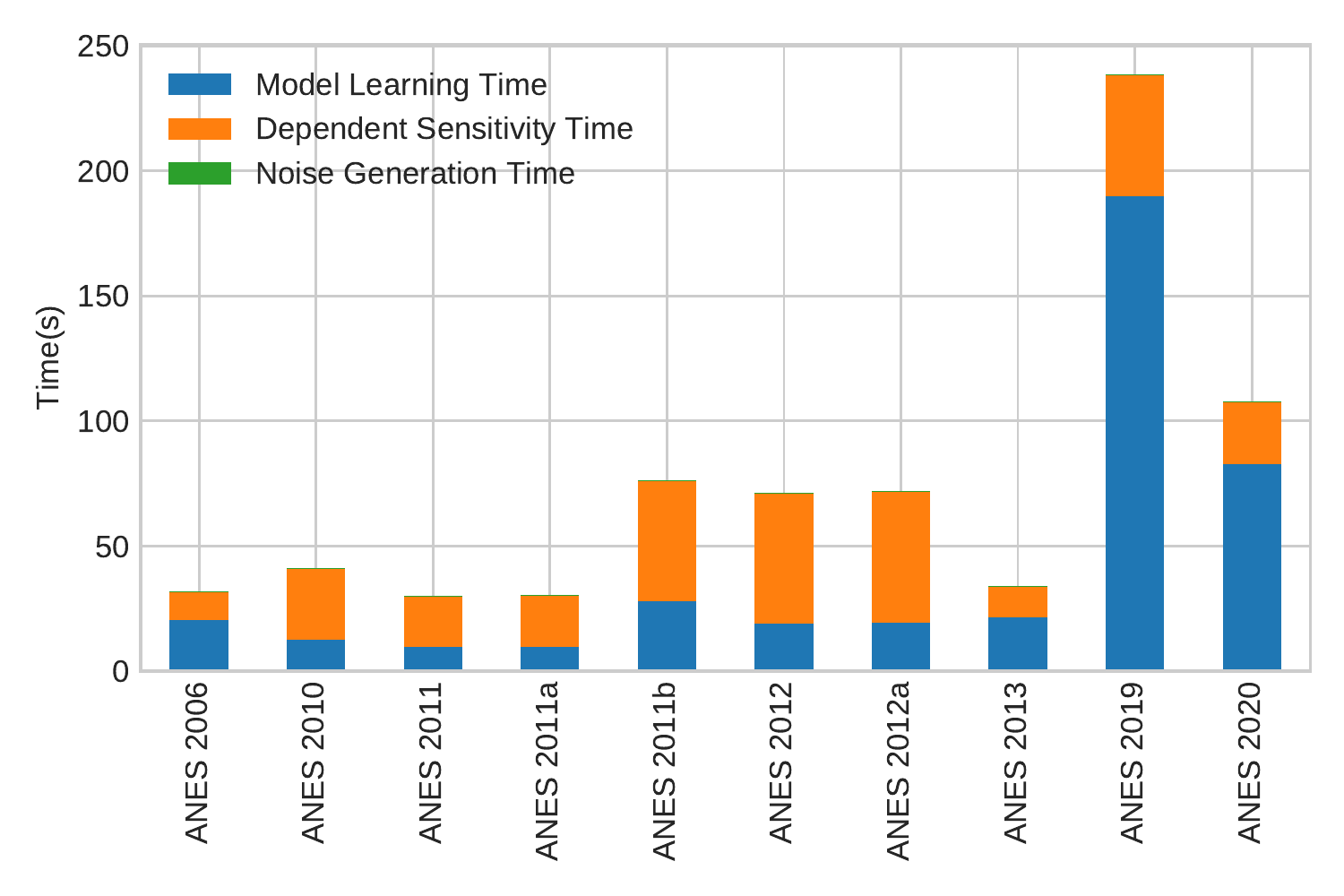}
    \\
    $k=10$ & $k=15$\\
  \end{tabular}
  \caption{Scalability results: per-component running time for the
    Tabular-DDP algorithm on each dataset. Note the difference in $y$
    axis scales. Model learning time increases exponentially with the
    value of $k$, and dominates for larger values of $k$.}
  \label{fig:scalability}
\end{figure*}

Our empirical evaluation seeks to answer two questions:
\begin{enumerate}
\item \textbf{Q1: Accuracy.} How does the accuracy of the \ourmech compare
  to the Laplace mechanism?
\item \textbf{Q2: Scalability.} How does the size and dimensionality
  of the dataset impact the running time of the \ourmech?
\end{enumerate}

To answer the first question, we evaluated the accuracy of the \ourmech for computing summary statistics for 9 survey datasets released by the American National Election Studies. The results suggest that \ourmech can significantly increase accuracy over the Laplace mechanism for these real-world datasets. To answer the second question, we measured running time for each component of \ourmech; the results suggest that \ourmech scales to realistic datasets, and that the primary scalability challenge comes from learning the Bayesian network from the data.

\paragraph{Datasets.}
Our datasets were drawn primarily from the American National Election Studies (ANES) database. Each dataset included columns that corresponded to the answers for questions in the metadata datasheet. Questions that were multiple choice were frequently designed by indexed characters (for example, in the dataset ANES 2011, multiple choice question responses are represented as sequential columns ("c3c1", "c3c2", where each possible answer is indicated by a number and choice, or otherwise, for that column "-1. Inapplicable, legitimate skip" ), and if there was branch logic for indexed questions, the numeric values were sentinel values. 

\paragraph{Methodology.}
We compared \ourmech to the Laplace mechanism, which provides $\epsilon$-differential privacy and assumes that attributes in each individual record may be completely correlated with one another. To simulate the computation of summary statistics for each survey, we ran a histogram query on each column of the survey results (i.e. we queried the count of each response category for each question of the survey).

\subsection{Experiment 1: Accuracy}


\paragraph{Experiment Setup.}
Our first experiment examines the accuracy of the \ourmech by comparing it to the standard Laplace mechanism. We ran 100 trials for each experiment, and report $L2$ error. We used $\epsilon \in \{0.1, 1, 10\}$ for both mechanisms.

\paragraph{Results.}
The results for $\epsilon = 0.1$ appear in Figure~\ref{fig:all_results}. Additional results for other values of $\epsilon$ appear in Figure~\ref{fig:all_results_appendix} in the Appendix, and are consistent with these. We set $k =10$ (i.e. 10 columns per ``chunk'' of the dataset, so that each Bayesian network covers 10 columns). The results show that the \ourmech consistently outperforms the Laplace mechanism in terms of accuracy at a given level of privacy.

Figure~\ref{fig:k_results_results} shows accuracy results for various values of the chunk size $k$. The results suggest that the accuracy advantage of the \ourmech over the Laplace mechanism increases as $k$ increases; when $k=5$, for example, the accuracy advantage of the \ourmech is fairly small, and it is much larger when $k=15$. These results match our expectations about the \ourmech: as $k$ increases, the \ourmech takes better advantage of the partiality of correlations between attributes.

\subsection{Experiment 2: Scalability}

\paragraph{Experiment Setup.}
Our second experiment measures running time of \ourmech to determine whether or not it can scale to realistic datasets. We instrumented our implementation to separately measure the running time of (1) learning the Bayesian network from the data, (2) calculating the dependent sensitivity, and (3) generating the noise samples themselves. We ran \ourmech on the same datasets and recorded the running time of each component; we performed 5 trials and report the average running time of each component. We set $k =10$ (i.e. 10 columns per ``chunk'' of the dataset, so that each Bayesian network covers 10 columns).

\paragraph{Results.}
The results appear in Figure~\ref{fig:scalability}, and suggest that \ourmech is capable of scaling to realistic datasets like the ANES surveys we considered. The running time for \ourmech in this experiment is dominated by the time to calculate dependent sensitivity based on the Bayesian network associated with the target columns. Running time was higher for surveys with more questions (e.g. the ANES 2012 and 2019 surveys, which had more columns than other datasets). For all of the datasets we considered, when $k=10$, \ourmech was able to compute summary statistics for all columns in about 10 seconds or less.

For small values of $k$, the running time is dominated by the time taken to calculate dependent sensitivity. However, as $k$ increases, the model learning time quickly dominates the total time, due to the fundamental scalability challenges of learning models over many attributes. The running time for the ANES 2020 survey increases 10x---from about 10 seconds to over 100 seconds---when $k$ increases from 10 to 15.


\subsection{Discussion}

Based on the results of our experiments, we answer the original research questions as follows. \textbf{(1)}: for the survey data we studied, the accuracy of the \ourmech improves on the Laplace mechanism---when $k \geq 10$, the improvement is often 2x or more. \textbf{(2)}: the \ourmech is slower than the Laplace mechanism, but for $k \leq 10$, it scales easily to realistic survey datasets with hundreds of columns and thousands of responses.

Our experimental results clearly demonstrate the tradeoff between running time and accuracy in the \ourmech: accuracy increases with larger values of $k$, but running time also increases (exponentially!). Fortunately, the results suggest that significant accuracy gains can be achieved with small enough values of $k$ that running time is reasonable. More scalable approaches for learning Bayesian networks may allow increasing $k$ further, and thus improving accuracy even more.

\section{Related Work}

A significant amount of previous work has considered the privacy implications of correlations within sensitive data. The most general framework for formalizing privacy while taking correlations into account is Pufferfish privacy~\cite{song2017pufferfish}, introduced earlier. The Pufferfish framework allows specifying any model of correlations in the underlying population as a probability distribution over possible datasets. Dependent differential privacy~\cite{liu2016dependence} can be defined as a particular variant of Pufferfish privacy. Our work builds on these definitions, providing a new mechanism that satisfies dependent differential privacy (and thus, Pufferfish privacy).

Many different mechanisms have been proposed for Pufferfish privacy; most are designed for a specific purpose where the correlations in the underlying data are known ahead of time to the analyst and have a specific structure. Many of these consider \emph{temporal} correlations---multiple data records contributed by the same individual over time---and model these correlations using Markov chains. Solutions have been proposed for social media settings~\cite{song2017pufferfish}, smart meter data~\cite{niu2019making, kessler2015deploying}, and web browsing data~\cite{liang2020pufferfish}. In contrast to these approaches, the \ourmech is designed to learn a general model of the underlying correlations from the data itself.

Recent work by Zhang et al.~\cite{zhang2022attribute} proposes Pufferfish mechanisms for \emph{attribute privacy}. This work uses similar techniques to ours, but has a different privacy goal: attribute privacy aims to prevent \emph{population-level} inferences about attributes of the dataset (for example, the distribution of race and gender in the original dataset). Our work, in contrast, aims to prevent inferences about \emph{individuals}.

Previous work has explored the application of differential privacy to protect privacy in survey data~\cite{d2015differential, evans2021statistically, evans2022differentially}. This work has focused on ensuring statistical validity and avoiding bias in the inferences made using differentially private statistics. Previous work in this area has applied well-known differential privacy mechanisms like the Laplace mechanism.

\section{Conclusion}

We have presented the \ourmech, a novel dependent differential privacy mechanism that can improve accuracy over the standard Laplace mechanism for high-dimensional statistics that are not completely correlated. We have shown how to apply the \ourmech to protect privacy in summary statistics for survey data; our experimental results show a significant improvement in accuracy compared to the standard Laplace mechanism in that setting.

\bibliographystyle{plain}
\bibliography{refs}

\begin{thebibliography}{10}

\bibitem{chan2011private}
T-H~Hubert Chan, Elaine Shi, and Dawn Song.
\newblock Private and continual release of statistics.
\newblock {\em ACM Transactions on Information and System Security (TISSEC)},
  14(3):1--24, 2011.

\bibitem{connors2019transparency}
Elizabeth~C Connors, Yanna Krupnikov, and John~Barry Ryan.
\newblock How transparency affects survey responses.
\newblock {\em Public Opinion Quarterly}, 83(S1):185--209, 2019.

\bibitem{d2015differential}
Vito D'Orazio, James Honaker, and Gary King.
\newblock Differential privacy for social science inference.
\newblock {\em Sloan Foundation Economics Research Paper}, (2676160), 2015.

\bibitem{dwork2006calibrating}
Cynthia Dwork, Frank McSherry, Kobbi Nissim, and Adam Smith.
\newblock Calibrating noise to sensitivity in private data analysis.
\newblock In {\em Theory of cryptography conference}, pages 265--284. Springer,
  2006.

\bibitem{dwork2014algorithmic}
Cynthia Dwork, Aaron Roth, et~al.
\newblock The algorithmic foundations of differential privacy.
\newblock {\em Foundations and Trends{\textregistered} in Theoretical Computer
  Science}, 9(3--4):211--407, 2014.

\bibitem{evans2021statistically}
Georgina Evans and Gary King.
\newblock Statistically valid inferences from differentially private data
  releases, with application to the facebook urls dataset.
\newblock {\em Political Analysis}, pages 1--21, 2021.

\bibitem{evans2022differentially}
Georgina Evans, Gary King, Adam~D Smith, and A~Thankurta.
\newblock Differentially private survey research.
\newblock {\em American Journal of Political Science}, 2022.

\bibitem{henriksen2016re}
Jane Henriksen-Bulmer and Sheridan Jeary.
\newblock Re-identification attacks—a systematic literature review.
\newblock {\em International Journal of Information Management},
  36(6):1184--1192, 2016.

\bibitem{kasiviswanathan2013analyzing}
Shiva~Prasad Kasiviswanathan, Kobbi Nissim, Sofya Raskhodnikova, and Adam
  Smith.
\newblock Analyzing graphs with node differential privacy.
\newblock In {\em Theory of Cryptography Conference}, pages 457--476. Springer,
  2013.

\bibitem{kessler2015deploying}
Stephan Kessler, Erik Buchmann, and Klemens B{\"o}hm.
\newblock Deploying and evaluating pufferfish privacy for smart meter data.
\newblock In {\em 2015 IEEE 12th Intl Conf on Ubiquitous Intelligence and
  Computing and 2015 IEEE 12th Intl Conf on Autonomic and Trusted Computing and
  2015 IEEE 15th Intl Conf on Scalable Computing and Communications and Its
  Associated Workshops (UIC-ATC-ScalCom)}, pages 229--238. IEEE, 2015.

\bibitem{liang2020pufferfish}
Wenjuan Liang, Hong Chen, Ruixuan Liu, Yuncheng Wu, and Cuiping Li.
\newblock A pufferfish privacy mechanism for monitoring web browsing behavior
  under temporal correlations.
\newblock {\em Computers \& Security}, 92:101754, 2020.

\bibitem{liu2016dependence}
Changchang Liu, Supriyo Chakraborty, and Prateek Mittal.
\newblock Dependence makes you vulnerable: Differential privacy under dependent
  tuples.
\newblock In {\em NDSS}, volume~16, pages 21--24, 2016.

\bibitem{niu2019making}
Chaoyue Niu, Zhenzhe Zheng, Shaojie Tang, Xiaofeng Gao, and Fan Wu.
\newblock Making big money from small sensors: Trading time-series data under
  pufferfish privacy.
\newblock In {\em IEEE INFOCOM 2019-IEEE Conference on Computer
  Communications}, pages 568--576. IEEE, 2019.

\bibitem{plutzer2019privacy}
Eric Plutzer.
\newblock Privacy, sensitive questions, and informed consent: Their impacts on
  total survey error, and the future of survey research.
\newblock {\em Public Opinion Quarterly}, 83(S1):169--184, 2019.

\bibitem{scanagatta2019survey}
Mauro Scanagatta, Antonio Salmer{\'o}n, and Fabio Stella.
\newblock A survey on bayesian network structure learning from data.
\newblock {\em Progress in Artificial Intelligence}, 8(4):425--439, 2019.

\bibitem{song2017pufferfish}
Shuang Song, Yizhen Wang, and Kamalika Chaudhuri.
\newblock Pufferfish privacy mechanisms for correlated data.
\newblock In {\em Proceedings of the 2017 ACM International Conference on
  Management of Data}, pages 1291--1306, 2017.

\bibitem{sturgis2021demise}
Patrick Sturgis and Rebekah Luff.
\newblock The demise of the survey? a research note on trends in the use of
  survey data in the social sciences, 1939 to 2015.
\newblock {\em International Journal of Social Research Methodology},
  24(6):691--696, 2021.

\bibitem{tsamardinos2006max}
Ioannis Tsamardinos, Laura~E Brown, and Constantin~F Aliferis.
\newblock The max-min hill-climbing bayesian network structure learning
  algorithm.
\newblock {\em Machine learning}, 65(1):31--78, 2006.

\bibitem{zhang2017privbayes}
Jun Zhang, Graham Cormode, Cecilia~M Procopiuc, Divesh Srivastava, and Xiaokui
  Xiao.
\newblock Privbayes: Private data release via bayesian networks.
\newblock {\em ACM Transactions on Database Systems (TODS)}, 42(4):1--41, 2017.

\bibitem{zhang2022attribute}
Wanrong Zhang, Olga Ohrimenko, and Rachel Cummings.
\newblock Attribute privacy: Framework and mechanisms.
\newblock In {\em 2022 ACM Conference on Fairness, Accountability, and
  Transparency}, pages 757--766, 2022.

\end{thebibliography}

\section*{Appendix}
See Figure~\ref{fig:all_results_appendix} for additional experimental results.

\begin{figure*}
  \centering
  
  \hspace*{-18mm}\includegraphics[width=1.2\textwidth]{figures/accuracy_results_0.1.pdf}

  \textbf{$\epsilon=0.1$}

  \hspace*{-18mm}\includegraphics[width=1.2\textwidth]{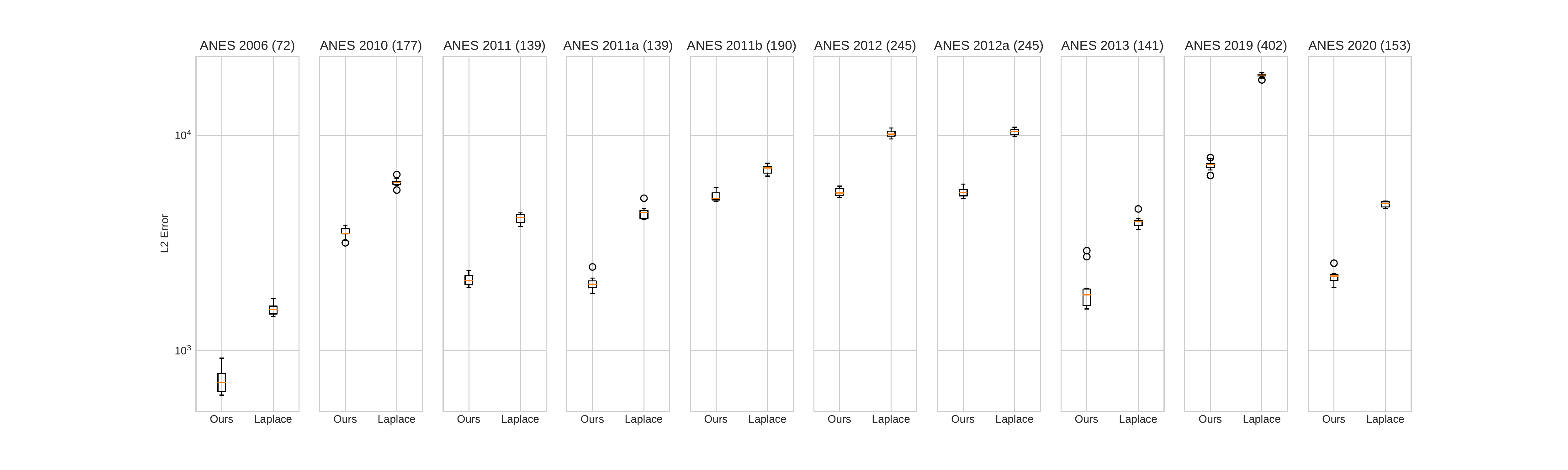}

  \textbf{$\epsilon=1$}

    \hspace*{-18mm}\includegraphics[width=1.2\textwidth]{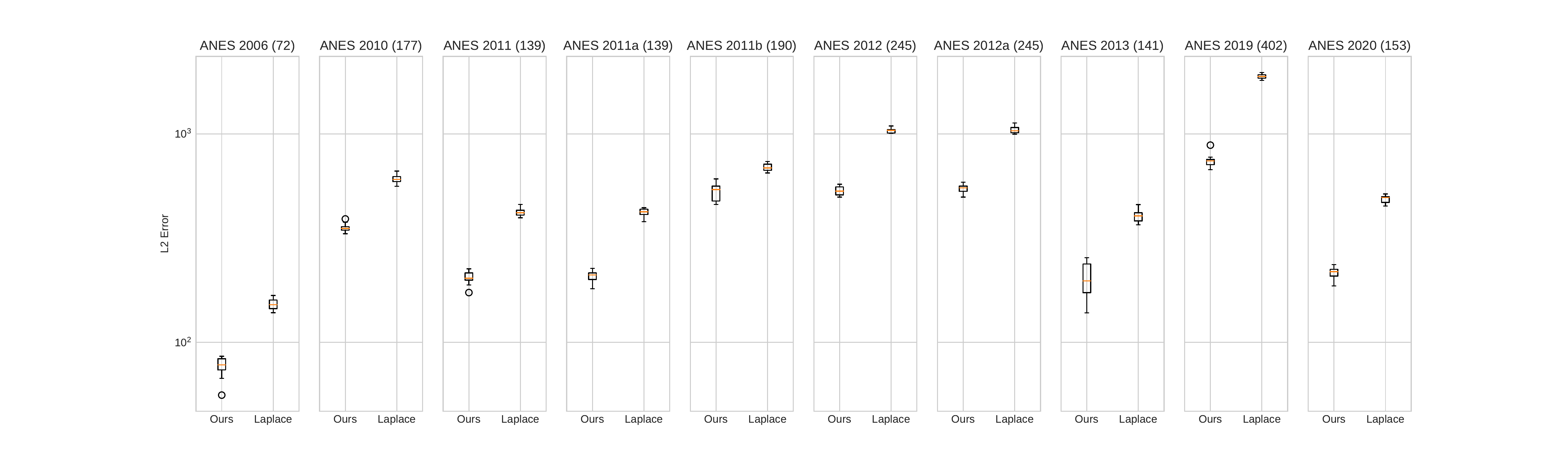}

  \textbf{$\epsilon=10$}

  \caption{Accuracy results: comparison between the Tabular-DDP
    mechanism and the Laplace mechanism for $\epsilon \in \{0.1, 1,
    10\}$.}
  \label{fig:all_results_appendix}
\end{figure*}

\end{document}